\documentclass[12pt,twoside]{article}

 \usepackage{float}
\usepackage{graphicx}
\usepackage{epstopdf}
\usepackage{graphicx}
\usepackage{epic}
\usepackage{multirow}
\usepackage{pst-poly}  
\usepackage{pst-plot}  
\usepackage{pst-poly}  
\usepackage{tikz}
\usepackage{xcolor}
\usetikzlibrary{arrows,shapes,chains}

\renewcommand{\paragraph}{\roman{paragraph}}
\usepackage[a4paper]{geometry}
\makeatletter
\renewcommand\title[1]{\gdef\@title{\reset@font\Large\bfseries #1}}
\renewcommand\section{\@startsection {section}{1}{\z@}%
                                   {-3.5ex \@plus -1ex \@minus -.2ex}%
                                   {2.3ex \@plus.2ex}%
                                   {\normalfont\large\bfseries}}
\renewcommand\subsection{\@startsection{subsection}{2}{\z@}%
                                     {-3ex\@plus -1ex \@minus -.2ex}%
                                     {1.5ex \@plus .2ex}%
                                     {\normalfont\normalsize\bfseries}}
\renewcommand\subsubsection{\@startsection{subsubsection}{3}{\z@}%
                                     {-2.5ex\@plus -1ex \@minus -.2ex}%
                                     {1.5ex \@plus .2ex}%
                                     {\normalfont\normalsize\bfseries}}

\def\@runningauthor{}\newcommand{\runningauthor}[1]{\def\runningauthor{#1}}
\def\@runningtitle{}\newcommand{\runningtitle}[1]{\def\runningtitle{#1}}

\renewcommand{\ps@plain}{%
\renewcommand{\@evenhead}{\footnotesize\scshape \hfill\runningauthor\hfill}
\renewcommand{\@oddhead}{\footnotesize\scshape \hfill\runningtitle\hfill}}

\newcommand{\F}{\mathbb{F}}
\newcommand{\x}{\mathbf{x}}
\newcommand{\A}{\mathbf{a}}

\newcommand {\ccc}{{\mathbf{c}}}
\newcommand {\dd}{\mathbf{d}}
\newcommand {\0}{\mathbf{0}}
\g@addto@macro\bfseries{\boldmath}

\makeatother

\usepackage{amsthm,amsmath,amssymb}
\usepackage{cite}
\usepackage{graphicx}

\usepackage[colorlinks=true,citecolor=black,linkcolor=black,urlcolor=blue]{hyperref}

\theoremstyle{plain}
\newtheorem{theorem}{Theorem}
\newtheorem{lemma}[theorem]{Lemma}

\newtheorem{proposition}[theorem]{Proposition}

\theoremstyle{definition}

\newtheorem{example}[theorem]{Example}
\newtheorem{conjecture}[theorem]{Conjecture}

\theoremstyle{remark}
\newtheorem{remark}[theorem]{Remark}

\title{The connections among Hamming metric, $b$-symbol metric, and $r$-th generalized Hamming metric
}

\runningtitle{The connections among Hamming metric, $b$-symbol metric, and $r$-th generalized Hamming metric}

\author{Minjia Shi\thanks{School of Mathematical Sciences, Anhui University, Hefei, China. E-mail: smjwcl.good@163.com}
\and Hongwei Zhu \thanks{ School of Mathematical Sciences, Anhui University, Hefei, China. E-mail: zhwgood66@163.com}
\and Tor Helleseth\thanks{Department of Informatics, University of Bergen, Bergen, Norway. E-mail: tor.helleseth@uib.no}
}

\runningauthor{}

\date{}

\begin{document}

\maketitle

\thispagestyle{empty}

\begin{abstract}
The $r$-th generalized Hamming metric and the $b$-symbol metric are two different generalizations of Hamming metric. The former is used on the wire-tap channel of Type II, and the latter is motivated by the limitations of the reading process in high-density data storage systems and applied to a read channel that outputs overlapping symbols. In this paper, we study the connections among the three metrics (that is, Hamming metric, $b$-symbol metric, and $r$-th generalized Hamming metric) mentioned above and give a conjecture about the $b$-symbol Griesmer Bound for cyclic codes. 
\end{abstract}

\section{Introduction}
The concept of $r$-th generalized Hamming metric first appeared in the 1970s and was proposed by Helleseth, Kl{\o}ve and Mykkeltveit \cite{Hell2,K1}. Until 1991,
in Wei's research \cite{Wei} on wire-tap channel of Type II, Wei mentioned this concept again and provided a series of excellent conclusions. Subsequently, many researchers studied the weight hierarchy of several series of linear codes (e.g. RM codes \cite{Hei1,Wei}, BCH codes \cite{Cheng1,Geer1,Geer2}, trace codes \cite{Stich}, cyclic codes \cite{Feng1,Janwa}, etc.). The bounds, asymptotic behaviour, and the duality under $r$-th generalized Hamming metric have been considered in \cite{Ash,Hell3,HKY,K1,Tsfa,Wei}. In addition to its applications in wire-tap channels of type II, the $r$-th generalized Hamming metric is also used to address the $t$-resilient functions and trellis or branch complexity of linear codes \cite{Tsfa}.

The $b$-symbol metric is another generalization of the Hamming metric that has been proposed by Cassuto et al. \cite{CB1,CB} in recent years. It differs from $r$-th generalized Hamming metric in that its research motivation is not derived from data storage or cryptography but from other domains such as molecular biology and chemistry. In these domains, the information redundancy is so low that the only effective way to combat errors is to transmit the same message over and over again (overlapping symbols). Although in practical applications,  consecutive symbols may affect each other, in the traditional read channel, people always assume that the adjacent symbols are individual. However, with the development of the high-density data storage technologies, this is no longer a reasonable assumption, and symbols are faced with the need to be read repeatedly (since the bit size at high-densities is small, it is hard to read the individual bits). This is why we have to pay attention to the $b$-symbol channel, which is a read channel suitable for the output of overlapping $b$-symbols. Under this new metric (or paradigm), errors are no longer single symbolic errors but $b$-symbolic errors.
At present, the research progress of the $b$-symbol metric includes the bounds of codes (e.g. $b$-symbol Sphere Packing  Bound \cite{CB1,CB}, $b$-symbol Singleton Bound \cite{C+,DTG}, $b$-symbol Linear Programming Bound \cite{Eli}, $b$-symbol asymptotic bound \cite{CL}, etc), the decoding and the constructions. The research on the codes that reach the $b$-symbol Singleton Bound (such codes are called $b$-symbol MDS codes) is a hot topic, and a lot of relevant research progress has been achieved \cite{C+,C+1,CLL,DGZ,DTG,KZL}. It is very difficult to determine the $b$-symbol weight distribution or the minimum $b$-symbol distance of linear codes. Nevertheless, in some special cases, the $b$-symbol weight distributions are determined \cite{SOS,Z,ZHW}.

These two metrics are widely concerned by researchers because they are generalizations of the Hamming metric.
Our motivation is to investigate the connections and differences between these two metrics. Since the $r$-th generalized Hamming metric has a longer history than the $b$-symbol metric, we can refer to the research progress of $r$-th generalized Hamming metric when we consider the $b$-symbol metric. In this paper, we first establish the connection between the Hamming metric and the $b$-symbol metric. Although the connection has been considered in \cite{Yaa,Yaa1}, we get better results (e.g., Theorem 4 is a generalization of Lemma 1 in \cite{Yaa,Yaa1}, and Theorem 6 is a generalization of Proposition 2 in \cite{Yaa,Yaa1}). Subsequently, We compare the same linear codes under the $b$-symbol metric and $r$-th generalized Hamming metric.
 When $C$ is cyclic, we prove that $d_b(C)\geq \dd_b(C)$, where $d_b(C)$ and $\dd_b(C)$ denote the minimum $b$-symbol weight and the minimum $b$-th generalized Hamming weight of $C$, respectively. In fact, the two metrics have a lot in common when $C$ is cyclic.
 We also propose a conjecture on the $b$-symbol Griesmer Bound for cyclic codes. 

The rest of the paper is organized as follows. In Section II, we introduce some notations and definitions. In Section III, we show the connections among Hamming metric, $b$-symbol metric, and $r$-th generalized Hamming metric. Section IV concludes the paper. 
\section{Preliminaries}
Throughout this paper, we assume  and fix the following:
\begin{itemize}
  \item $\F_q:$ finite field with $q$ elements.
  \item $\F_q^*=\F_q\backslash\{0\}$.
  \item $p={\rm Char}(\F_q).$
  \item $\x,\mathbf{y}$ are two vectors which belong to $\F_q^n$.
\end{itemize}
\subsection{Hamming metric}
\begin{itemize}
  \item Hamming weight $w_H(\mathbf{x})$: the number of nonzero coordinates in $\mathbf{x}.$
  \item Hamming distance $d_H(\mathbf{x},\mathbf{y})$: the number of coordinates in which $\mathbf{x}$ and $\mathbf{y}$ differ.
\end{itemize}
\subsection{The $b$-symbol metric}
Let $b$ be a positive integer with $1\leq b\leq n.$
\begin{itemize}
  \item $b$-symbol weight $w_b(\x)$: the Hamming weight of $\pi_b(\x)$, where $\pi_b(\x)\in (\F_q^b)^n$ and
$$\pi_b(\x)=((x_0,\ldots,x_{b-1}),(x_{1},\ldots,x_{b}),
\cdots,(x_{n-1},\ldots,x_{b+n-2({\rm mod}~n)})).$$
  \item $b$-symbol distance $d_b(\x,\mathbf{y})$: $d_b(\x,\mathbf{y})=w_b(\x-\mathbf{y}).$
\end{itemize}

When $b=1$, $w_1(\x)=w_H(\x)$ and $d_1(\x,\mathbf{y})=d_H(\x,\mathbf{y})$. For convenience, we use $w_1(\x)$ and $d_1(\x,\mathbf{y})$ to represent $w_H(\x)$ and $d_H(\x,\mathbf{y})$, respectively.
For $\eta\in \F_q$ and $\lambda, \xi\in \F_q\setminus\{\eta\}$, if $\eta,\underbrace{\lambda,\ldots,\lambda}_m, \xi$ appears in the sequence $\mathbf{\overline{a}}$, then we say that $\underbrace{\lambda,\ldots,\lambda}_m$ is a run of $\lambda$'s of length $m$. Let $\mathbf{0}_i=(\alpha,\underbrace{0,\ldots,0}_i,\beta)$, where $\alpha, \beta\in\F_q^*$.
For any vector $\mathbf{a}=(a_0,a_1,\ldots,a_{n-1})\in \F_q^n$, we define a circumferential vector $cir(\mathbf{a})$ as follows:
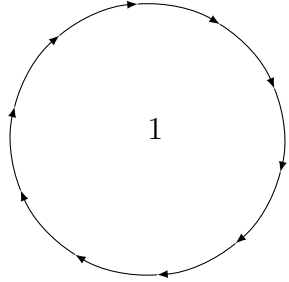
\begin{center}\begin{tikzpicture}
$\def \radius {1.8cm}
\def\startDegree{50}
\def\n{10}
\foreach \s in {1,...,\n}
{
  \draw[<-, >=latex] ({360/\n * (\s - 1)-\startDegree}:\radius)
    arc ({360/\n * (\s - 1)-\startDegree}:{360/\n * (\s - 1)-\startDegree+360/\n}:\radius);
}
$,
\put(-12,56){\fontsize{12}{2.5}\selectfont $ a_0$}
\put(25,50){\fontsize{12}{2.5}\selectfont $ a_1 $}
\put(50,22){\fontsize{12}{2.5}\selectfont $ a_2 $}
\put(-6,-42){\fontsize{12}{2.5}\selectfont $ \ldots $}
\put(-80,12){\fontsize{12}{2.5}\selectfont $ a_{n-2} $}
\put(-60,42){\fontsize{12}{2.5}\selectfont $ a_{n-1} $}
\end{tikzpicture}
\end{center}
 and $\Psi(\mathbf{0}_i)$ denotes the number of occurrences of $\mathbf{0}_i$ on the circumferential vector $cir(\mathbf{a})$.
 The $\mathbf{0}$'s run distribution of $\mathbf{a}$ is defined by
$\Psi(\mathbf{a})=\{\Psi(\mathbf{0}_1), \Psi(\mathbf{0}_{2}),\ldots, \Psi(\mathbf{0}_{n})\}.$
\begin{example}
Let $\mathbf{a}=(01001000100)$. Then the $\mathbf{0}$'s run distribution of $\mathbf{a}$ is
$$\Psi(\mathbf{a})=\{\Psi(\mathbf{0}_2)=1,\Psi(\mathbf{0}_3)=2,\Psi(\mathbf{0}_i)=0, i\neq 2,3\}.$$
\end{example}

For any vector $\mathbf{c}=(c_0,c_1,\ldots,c_{n-1})\in\F_q^n$, by the definition of $b$-symbol weight, then we have
$$w_b(\mathbf{c})=n-|\{i|c_i=c_{i+1}=\cdots=c_{i+b-1}=0,0\leq i\leq n-1\}|.$$
If $\Psi(\ccc)$ is given, we have the following formula to calculate the $b$-symbol weight of $\ccc$.

 \begin{theorem}\label{wbformula}
 For any vector $\mathbf{c}=(c_0,c_1,\ldots,c_{n-1})\in\F_q^n$, we have
\begin{equation}\label{formula}
  w_b(\ccc)=n-\sum\limits_{i=b}^{n-1}(i-b+1)\cdot \Psi(\mathbf{0}_i).
\end{equation}
\end{theorem}
\begin{proof}
By the definition of the $b$-symbol metric, the zero coordinate of $\ccc$ is $b$ cyclic consecutive zeros.
For a  $\0_i=(\alpha,\underbrace{0,\ldots,0}_i,\beta)$, there are $i-b+1$ zero coordinates if $i\geq b$ and there is no zero coordinate in $\0_i$ if $i<b$. Counting the number of the zero coordinates, we obtain the desired result.
\end{proof}

From Theorem \ref{wbformula}, one can induce the following lemma directly.
\begin{lemma}\label{Lem1}
Let $\mathbf{x}$ and $ \mathbf{y}$ be two vectors that belong to $\F_q^n$, then $w_{b}(\mathbf{x})=w_{b}(\mathbf{y})$ if $\Psi(\mathbf{x})=\Psi(\mathbf{y})$, where $1\leq b\le n$.
\end{lemma}
\begin{remark}\label{r4}
A linear code $C$ is called a single $\mathbf{0}$'s run code if for any $\ccc, \ccc^{\prime}\in C\backslash\{\mathbf{0}\}$, $\Psi(\ccc)=\Psi(\ccc^{\prime}).$ For example, the simplex code with parameters $[\frac{q^k-1}{q-1},k]$ is a single $\mathbf{0}$'s run code with the only $\mathbf{0}$'s run distribution
\begin{equation*}
  \Psi(\mathbf{0}_i)=
\left\{
  \begin{array}{ll}
    (q-1)q^{k-2-i}, & \hbox{$1\leq i\leq k-2$;} \\
    1, & \hbox{$i=k-1$;} \\
    0, & \hbox{$i\geq k$.}
  \end{array}
\right.
\end{equation*}
Its $\mathbf{0}$'s run distribution can be deduced from \cite[Page.123, R-2]{Golomb}. By Theorem \ref{wbformula}, we have the minimum $b$-symbol weight of the simplex code is $\frac{(q^b-1)q^k}{q^b(q-1)}.$ Obviously, the simplex code is also a single $b$-symbol weight linear code.
\end{remark}
The single $\mathbf{0}$'s run codes are of interest in their own right. The following result gives a class of single $\mathbf{0}$'s run codes and their parameters.
\begin{theorem}\label{single1}Let $\Delta$ be a factor of $q^k-1$.
Let $n=\frac{q^k-1}{\Delta}$ and $C$ be an irreducible cyclic code over $\F_q$ with parity-check polynomial $h(x)$ of degree $\deg(h(x))=k$ and period (or order) ${\rm per}(h(x))=n$.
If $\Delta|q-1$ and $\gcd(n,\Delta)=1$, then $C$ is a single $\mathbf{0}$'s run code with parameters $[n,k,d_b(C)=
\frac{(q^b-1)q^k}{q^b\Delta}]_q$ where $1\leq b\leq k-1$.
\end{theorem}
\begin{proof}
Let $\bar{\mathbf{a}}=(a_0,a_1,a_2,\ldots)$ be a nonzero sequence. We use $G(h^*(x))$ to denote the set consisting of all sequences with $h^*(\tau)\bar{\mathbf{a}}=0$ where $\tau$ denotes the left shift operator. Let $\alpha$ be a primitive element of $\F_{q}$.
If $\gcd({\rm per}(h(x)),\Delta)=1$, we claim that the codewords $\mathbf{x}$ and $\mathbf{y}$ are cyclically shift distinct if there exists $\delta\in \langle \alpha^{\frac{q-1}{\Delta}}\rangle\setminus\{1\}\subseteq\F_q^*$ such that $\mathbf{x}=\delta \cdot \mathbf{y}$.
The claim can be proved by showing that for a non-zero sequence $\overline{\A}\in G(h^*(x))$,
it is impossible for both states $\xi_i$ and $\xi_j$ of $\overline{\A}$
 such that $\xi_i=\delta\cdot\xi_j$ where $i\neq j$.
Assume that there are two states $\xi_i$ and $\xi_j$ of $\overline{\A}$ such that $\xi_i=\delta\cdot\xi_j$ for some $\delta\in\langle \alpha^{\frac{q-1}{\Delta}}\rangle\setminus\{1\}.$ As it can be seen from the recurrence relation,
$$\xi_i=\delta\cdot\xi_j=\delta^2\cdot\xi_{j+j-i}=\cdots =\delta^{\Delta+1}\cdot\xi_{j+\Delta(j-i)}=\delta\cdot\xi_{j
+\Delta(j-i)},$$
then ${\rm per}(h(x))|\Delta(j-i)$. Since $\gcd({\rm per}(h(x)),\Delta)=1$ and ${\rm per}(h(x))|(j-i)$, this leads to a contradiction. Therefore, the non-zero codewords $\x,\alpha^{\frac{q-1}{\Delta}}\x,(\alpha^{\frac{q-1}{\Delta}})^2\x,
\ldots,(\alpha^{\frac{q-1}{\Delta}})^{\Delta-1}\x$ are from distinct cycles. Since there are $\Delta$ cycles (we ignore the cycle corresponding to zero sequence), $C$ has one non-zero $b$-symbol weight.

We claim that for any non-zero codeword $\mathbf{c}=(c_0,c_1,\ldots,c_{n-1})\in C$, the $\mathbf{0}$'s run distribution of $\mathbf{c}$ is
\begin{equation*}
  \Psi(\mathbf{0}_i)=\left\{
                          \begin{array}{ll}
                            \frac{(q-1)^2}{\Delta}q^{k-2-i}, & \hbox{$1\leq i\leq k-2$;} \\
                            \frac{q-1}{\Delta}, & \hbox{$i=k-1$;} \\
                            0, & \hbox{$i\geq k$.}
                          \end{array}
                        \right.
\end{equation*}
There is no $\mathbf{0}$'s run with length greater than or equal to $k$ since $\overline{\A}$ is not a zero state. When $1\leq i\leq k-2$, we add $k-i-2$ coordinates $(b_1,b_2,\ldots,b_{k-i-2})$ behind $(a_1,0,\ldots,0,a_2)$, where $a_1,a_2\in \F_q^*$, $b_{j_1}\in\F_q, 1\leq j_1\leq k-i-2.$ There are $q^{k-i-2}(q-1)^2$ choices of $\{a_1,a_2,b_1,\ldots,b_{k-i-2}\}$. But there are only $\frac{q^{k-i-2}(q-1)^2}{\Delta}$ states which cover $\mathbf{0}_i$ can be obtained in the same non-zero sequence $\overline{\A}$. When $i=k-1$, let $\xi_1=(a,0,\ldots,0)$ be a state of a non-zero sequence $\overline{\A}$, where $a\in\F_q^*.$ According to the linear recursive relation, $\xi_2=(0,\ldots,0,b)$, where $b\in\F_q^*$ and $b$ is determined by $a$. Similarly, there are only $\frac{q-1}{\Delta}$ choices of $a$ such that $\xi_1$ is a state of $\overline{\A}.$

For any non-zero codeword $\mathbf{c}\in C$, it then follows from Theorem \ref{wbformula} that
\begin{eqnarray*}
  w_b(\mathbf{c}) &=& n-\sum\limits_{i=b}^{n-1}(i-b+1)\cdot \Psi(\mathbf{0}_i) \\
  ~ &=& \frac{q^k-1}{\Delta}-\frac{(q-1)^2}{\Delta}\cdot\sum\limits_{j_2=b+1}^{k-1}\sum\limits_{i_2=1}^{k-j_2}q^{i_2-1} -\frac{(q-1)(k-2)}{\Delta} \\
  ~ &=&  \frac{q-1}{\Delta}(q^{k-1}+\cdots+q^{k-b})=\frac{(q^b-1)q^{k}}{q^b\Delta}.
\end{eqnarray*}
This completes the proof.
\end{proof}
The conditions given by Lemma \ref{Lem1} are not necessary. Here we give a sufficient and necessary condition for the $b$-symbol weight of two vectors to be equal. To this end, we give a generalization of the support of a vector $\ccc$. We define the $b$-symbol support to be
$$\mathcal{I}_b(\ccc)=supp(\pi_b(\ccc))=\bigcup_{i=0}^{b-1} supp(\tau^i(\ccc)),$$
where $supp(\ccc)$ denotes the support of a vector $\ccc$ and $\tau$ denotes the left shift operator. It is easy to check that $w_b(\ccc)=|\mathcal{I}_b(\ccc)|.$
\begin{theorem}
Let $\x$ and $\mathbf{y}$ be two vectors that belong to $\F_q^n$, then $w_b(\x)=w_b(\mathbf{y})$ if and only if $|\mathcal{I}_b(\x)|=|\mathcal{I}_b(\mathbf{y})|.$
\end{theorem}
\begin{proof}
The desired result follows from the definition of $b$-symbol metric.
\end{proof}
\subsection{The $r$-th generalized Hamming metric}
Let $C$ be an $[n,k]_q$ linear code.
 For any subcode $D\subset C$ , then the support of $D$ is defined to be
 \begin{equation*}
   \chi(D)=\{i:0\leq i\leq n-1|c_i\neq0~ \hbox{for some} ~ (c_0,c_1,\ldots,c_{n-1})\in D\}.
 \end{equation*}
The $r$-th generalized Hamming weight of a code $C$ is the smallest support of an $r$-dimensional subcode of $C$.
To avoid confusion of notation, we use $\dd_r(C)$ for $r$-th generalized Hamming distance of $C$.

The set $\{\dd_r(C):1\leq r\leq k\}$ is called the weight hierarchy of $C$.
To distinguish it from the later definition, let us call it the generalized weight hierarchy in the sequel.

\section{The connections among Hamming metric, $b$-symbol metric, and $r$-th generalized Hamming metric}
The goal of this section is to show the connections among Hamming metric, $b$-symbol metric and $r$-th generalized Hamming metric. 
\subsection{Hamming metric and $b$-symbol metric}
Let $G_b(\mathbf{c})$ be the generator matrix of the code generated by $\mathbf{c}$ and its first $b$ cyclic shifts.
\begin{equation}\label{gbc}
  G_b(\mathbf{c})=\left(
  \begin{array}{cccccc}
    c_0 & c_1 & c_2 & \cdots & c_{n-2} & c_{n-1} \\
    c_1 & c_2 & c_3 & \cdots & c_{n-1} & c_0 \\
    c_2 & c_3 & c_4 & \cdots & c_0 & c_1 \\
    \vdots & \vdots & \vdots & \vdots & \vdots & \vdots \\
    c_{b-1} & c_b & c_{b+1} & \cdots & c_{b-3} & c_{b-2} \\
  \end{array}
\right).
\end{equation}
The following result shows the connection between the $b$-symbol weight and Hamming weight for any vector. It is a generalization of Lemma 1 in \cite{Yaa}.
\begin{theorem}\label{relation}
Let $\mathbf{c}\in \F_q^n$ and denote by $V_b(\mathbf{c})$ the vectors generated by all linear combinations of $\ccc$ and its first $b-1$ cyclic shifts $($i.e., generated by all linear combinations of $G_b(\ccc))$. Then
$$w_b(\mathbf{c})=\frac{1}{q^{b-1}(q-1)}\sum_{\mathbf{c}^{\prime}\in V_b(\mathbf{c})}w_1(\mathbf{c}^{\prime}).$$
\end{theorem}
\begin{proof}
The number of nonzero columns in the generator matrix of $V_b$ (its support size) is exactly $w_b(\mathbf{c})$. Hence, the sum of Hamming weight of codewords in $V_b(\mathbf{c})$ obey $$\sum_{\mathbf{c}^{\prime}\in V_b(\mathbf{c})} w_1(\mathbf{c}^{\prime})=q^{b-1}(q-1)w_b(\mathbf{c})$$ which proves the theorem.
\end{proof}
Notice that $|V_b(\ccc)|=q^b$ and                                                                                                                                                                                                                                                                                                                                                                                                                                                                                                                                                                                                                                                                                                                                                                                                                                                                                                                                                                                                                                                                                                                                                                                                                                                                                                                                                                                                                                                                                                                                                                                                                                                                                                                                                                                                                                                                                                                                                                                                                                                                                                                                                                                                                                                                                                                                                                                                                                                                                                                                                                                                                                                                                                                                                                                                                                                                                                                                                                                                                                                                                                                                                                                                                                                                                                                                           $V_b(\ccc)$ may be an multiset. For instance,
if we take $\ccc=(1010)\in\F_2^4$, then
\begin{eqnarray*}
  V_2(\ccc) &=& \{0000,1010,0101,1111\}; \\
  V_3(\ccc) &=& \{0000,1010,0101,1111,0000,1010,0101,1111\}.
\end{eqnarray*}
Observe that if the minimal polynomial of $\ccc=(c_0,c_1,\ldots,c_{n-1})$ has degree $\rho(\ccc)<b$, then $w_b(\ccc)=w_{\rho(\ccc)}(\ccc)$. This follows since $G_b(\ccc)$ contains the same row space as $G_{\rho(\ccc)}(\ccc)$ but each with multiplicity $q^{b-\rho(\ccc)}.$

In \cite{DTG,Yaa} the authors gave the connection between $w_b(\ccc)$ and $w_1(\ccc)$ for any vector $\ccc\in\F_q^n$ with $0<w_1(\ccc)\leq n-(b-1)$ as:
\begin{equation}\label{ineq}
  w_1(\ccc)+b-1\leq w_b(\ccc)\leq b\cdot w_1(\ccc).
\end{equation}
The following result generalizes Inequality (\ref{ineq}) to the general case and gives an interesting triangle inequality about the $b$-symbol metric. To this end, we need the following lemma.
\begin{lemma}\label{lemma5}
Let $\ccc$ be a vector that belongs to $\F_q^n$. Then
$$w_b(\ccc)\geq b\sum_{i=b}^{n-1}\Psi(\0_i).$$
\end{lemma}
\begin{proof}
The result follows since each nonzero element before the start of a run of $b$ $\0$'s or more belongs to $b$ nonzero $b$-tuples.
\end{proof}
\begin{theorem}\label{thm5}
$($Triangle inequality$)$ Let $\ccc\in\F_q^n$ be such that $0<w_b(\ccc)\leq n-m$ and $0< w_m(\ccc)\leq n-b$. Then we have
\begin{equation*}
  \max\{w_b(\ccc)+m,w_m(\ccc)+b\}\leq w_{b+m}(\ccc)\leq \min\{w_b(\ccc)+w_m(\ccc),n\}.
\end{equation*}
\end{theorem}
\begin{proof}
Let $\{\Psi(\mathbf{0}_1),\Psi(\mathbf{0}_2),\ldots,\Psi(\mathbf{0}_{n})\}$ be the $\mathbf{0}$'s run distribution of $\ccc.$
According to Theorem \ref{wbformula}, we have
\begin{equation}\label{A1}
  w_b(\ccc) = n-\sum_{i=b}^{n-1}(i-b+1)\cdot\Psi(\mathbf{0}_i);
\end{equation}
\begin{equation}\label{A2}
w_{m}(\ccc) =n-\sum_{i=m}^{n-1}(i-m+1)\cdot\Psi(\mathbf{0}_i);
\end{equation}
\begin{equation}\label{A3}
 ~~~~~ w_{b+m}(\ccc) = n-\sum_{i=b+m}^{n-1}(i-b-m+1)\cdot\Psi(\mathbf{0}_i).
\end{equation}
\begin{itemize}
  \item If $\sum_{i=b+m}^{n-1}\Psi(\mathbf{0}_i)=0$, then $w_{b+m}(\ccc)=n\geq w_b(\ccc)+m$;
  \item If $\sum_{i=b+m}^{n-1}\Psi(\mathbf{0}_i)\geq1$, then
  \begin{eqnarray*}
             w_{b+m}(\ccc)-w_b(\ccc) &=& \sum_{i=b}^{b+m-1}(i-b+1)\cdot\Psi(\mathbf{0}_i)+
             m\sum_{i=b+m}^{n-1}\Psi(\mathbf{0}_i)\geq m.
           \end{eqnarray*}
\end{itemize}
Therefore, we complete the proof of the inequality on the left-hand side.

Without loss of generality, assume that $b\leq m$.
In the light of Equations (\ref{A1}), (\ref{A2}) and (\ref{A3}), we obtain
\begin{eqnarray*}
  w_b(\ccc)+w_m(\ccc)-w_{b+m}(\ccc) &=& w_b(\ccc)+n-\sum_{i=m}
  ^{n-1}(i-m+1)\cdot\Psi(\mathbf{0}_i)\\
   ~&~&- n+\sum_{i=b+m}^{n-1}(i-b-m+1)\cdot\Psi(\mathbf{0}_i)\\
  ~ &=& w_b(\ccc)
  -b\sum_{i=m+b}^{n-1}\Psi(\mathbf{0}_i)-\sum_{i=m}^{m+b-1}
  (i-m+1)\cdot\Psi(\mathbf{0}_i)\\
  ~ &\geq&0~~(\hbox{by Lemma \ref{lemma5} and the assumption $b\leq m$}).
\end{eqnarray*}
Therefore, we get the desired results.
\end{proof}

\begin{remark}
According to Proposition 2.3 in \cite{DTG}, we have $w_{b+1}(\ccc)\geq w_b(\ccc)+1$ if $0<w_b(\ccc)<n$. This is a simpler way to prove the left-hand side of the above inequality.
\end{remark}

From Theorem \ref{thm5}, we have the following proposition.
\begin{proposition}
Let $\ccc\in\F_q^n$ and $0<w_b(\ccc)<n$. Then we have
\begin{equation*}
  \max\left\{w_{k_1}(\ccc)+b-k_1,\ldots,w_{k_t}(\ccc)+b-k_t\right\}
  \leq w_b(\ccc)\leq\min\left\{\sum_{i=1}^tw_{k_i}(\ccc),n\right\},
\end{equation*}
where $b=\sum_{i=1}^tk_i.$
\end{proposition}
\subsection{$b$-symbol metric and $r$-th generalized Hamming metric}
Let $G=[s_0,s_1,\cdots,s_{n-1}]$ be a generator matrix of an $[n,k]$ code $C$ with columns $s_i$, $i=0,1,\ldots,n-1$.
Let $U$ be a subspace of $\F_2^k$, $m(U)=\{i|s_i\in U\}$  and
$$F_{k,l}=\{U|\dim(U)=l\}.$$

\begin{lemma}{\rm\cite{HKY}}\label{Torlem1}
Let $C$ be an $[n,k]$ linear code over $\F_q$. Then
$$\dd_r(C)=n-\max\{m(U)|U\in F_{k,k-r}\}.$$
\end{lemma}
\begin{theorem}\label{thm9}
Let $C$ be an $[n,k]$ linear code over $\F_q$. Then
\begin{eqnarray}
  \dd_r(C) &=& \min\left\{\left.\frac{1}{q^{r-1}(q-1)}\sum_{\ccc\in R}w_1(\ccc)\right| R \hbox{ is an $r$-dimensional subspace of $C$}
  \right\}; \\
  d_b(C) &=& \min\left\{\left.\frac{1}{q^{b-1}(q-1)}\sum_{\ccc^{\prime}\in V_b(\ccc)}w_1(\ccc^{\prime})\right|\ccc\in C\right\}.
\end{eqnarray}
\end{theorem}
\begin{proof}
The desired result (7) follows from the definition of the $r$-th generalized Hamming weight and Lemma \ref{Torlem1}. The desired result (8) follows from Theorem \ref{relation}.
\end{proof}


In the case of $b$-symbol weight, we need to compute the summation of the Hamming weights of vectors in $V_b(\ccc)$ for any $\ccc$, 
while for the $r$-th generalized Hamming weight, this summation needs to be calculated for all $r$-dimensional subcodes.

The next goal is to give a very interesting connection between $d_b(C)$ and $\dd_r(C)$ if $C$ is cyclic. To this end, we need the following lemmas.
\begin{lemma}\label{wbdr}
Let $C$ be a cyclic code over $\F_q$ with parameters $[n,k]$ and $\ccc$ be a codeword of $C$. Let $b$ be a positive integer which not greater than $k$. If $G_b(\ccc)$ has rank $\rho(\ccc)$, then
$$w_b(\ccc)\geq \dd_{\rho(\ccc)}(C).$$
\end{lemma}
\begin{proof}
If $C$ is cyclic and $G_b(\ccc)$ has rank $\rho(\ccc)$, then the code $C^{\prime}$ generated by $G_b(\ccc)$ is a subcode of $C$ and $C^{\prime}$ is a $\rho(\ccc)$-dimensional subspace of $C$. By Theorem \ref{relation}, we have
\begin{eqnarray*}
  w_b(\ccc) &=& \frac{1}{q^{b-1}(q-1)}\sum_{\ccc^{\prime}\in V_b(\ccc)}
w_1(\ccc^{\prime}) \\
  ~         &=& \frac{q^{b-\rho(\ccc)}}{q^{b-1}(q-1)}\sum_{\ccc^{\prime}\in C^{\prime}}w_1(\ccc^{\prime})\\
  ~         &\geq& \dd_{\rho(\ccc)}(C).
\end{eqnarray*}
Therefore, we obtain the desired result.
\end{proof}
\begin{remark}
The reason why we need the restriction $b\leq k$ is to ensure $\rho(\ccc)\leq k$. From the definition of $r$-th generalized Hamming weight, $1\leq r\leq k$ and $C$ has to be linear. The definition of the $b$-symbol weight does not restrict $b\leq k$, but if $C$ is cyclic, we only consider the case $b\leq k$, for reasons that we will explain in Theorem \ref{Thm9} and Remark \ref{re23}.
\end{remark}
Assume that $\ccc$ is a codeword of $C$.
Let $\vartheta(\ccc)$ be the maximum $\mathbf{0}$'s run length of $cir(\ccc)$ and $\theta=\max\{\vartheta(\ccc)|\ccc\in C\backslash\{\mathbf{0}\}\}$. The parameter $\theta$ is called the maximum $\mathbf{0}$'s run length of $C$.
\begin{lemma}\label{new7}
Let $C$ be a linear code over $\F_q$ with minimum $b$-symbol weight $d_b(C)<n$. Let
$\overline{C}=\{\ccc|\ccc\in C \hbox{~and~} w_b(\ccc)=d_b(C)\}.$ Then $Rank(G_b(\ccc))=b$ for any $\ccc\in \overline{C}$.
\end{lemma}
\begin{proof}
We claim that $\vartheta(\ccc)\geq b$ for any $\ccc\in\overline{C}$. Assume that there exists a codeword $\ccc^{\prime}\in \overline{C}$ such that $\vartheta(\ccc)\leq b-1$. By Theorem \ref{wbformula}, $w_b(\ccc^{\prime})=n$, a contradiction. Since $\vartheta(\ccc)\geq b$,
for any $\ccc\in \overline{C}$, $\ccc$ is of the form
$$\ccc=(\ldots,\alpha,\underbrace{0,\ldots,0}_{\vartheta(\ccc)\geq b},\beta,\ldots),$$
 where $\alpha$ and $\beta$ belong to $\F_q^*.$
Then the matrix $G_b(\ccc)$ has a $b$ by $\vartheta(\ccc)+2$ submatrix of the form
$$G^{\prime}=\left(
               \begin{array}{ccccccc}
                 \alpha & 0 & 0 & \cdots & 0 & 0 & \beta \\
                 0 & 0 & 0 &\cdots &0 & \beta & \*  \\
                 \vdots & \vdots & \vdots & \vdots & \vdots & \vdots & \vdots \\
                 0 & \cdots & 0 & \beta & \* & \* & \* \\
               \end{array}
             \right)
_{b\times(\vartheta(\ccc)+2)}.
$$
The desired result follows from $Rank(G^{\prime})=b.$
\end{proof}
\begin{theorem}\label{new8}
Let $C$ be a cyclic code over $\F_q$ with parameters $[n,k]$ and $b$ be a positive integer which not greater than $k$. Then
$d_b(C)\geq \dd_b(C)$.
\end{theorem}
\begin{proof}
The desired result follows from Lemma \ref{wbdr} and Lemma \ref{new7}.
\end{proof}

Just as the weight hierarchy of $C$ under the $r$-th generalized Hamming metric, we define the weight hierarchy of $C$ under the $b$-symbol metric as follows:
\begin{itemize}
  \item $b$-symbol weight hierarchy of $C$: $\{d_1(C),d_2(C),\ldots,d_{\theta}(C),d_{\theta+1}(C),\ldots,d_{n}(C)\}.$
\end{itemize}
\begin{theorem}\label{thm10}
Let $C$ be an $[n,k]$ linear code of the maximum $\mathbf{0}$'s run length $\theta$ over $\F_q$. The following hold:
\begin{enumerate}
  \item $0<d_1(C)<d_2(C)<\cdots<d_{\theta}(C)<d_{\theta+1}(C)=\cdots =d_{n}(C)=n.$
  \item The $b$-symbol weight hierarchy of $C$ is the same as the $b$-symbol hierarchy of $CDP^i$, where $D$ is an $n$ by $n$ diagonal matrix,
\begin{equation*}
  P=\left(
      \begin{array}{ccccc}
        0 & 0 & 0 & \cdots & 1 \\
        1 & 0 & 0 & \cdots & 0 \\
        0 & 1 & 0 & \cdots & 0 \\
        \vdots & \vdots & \ddots & \ddots & \vdots \\
        0 & 0 & \cdots & 1 & 0 \\
      \end{array}
    \right)
  _{n\times n}
\end{equation*}
and $0\leq i\leq n-1$.
\end{enumerate}
\end{theorem}
\begin{proof}
Let $\ccc$ be a codeword in $C$ of minimum nonzero $b$-symbol weight with $2\leq b\leq \theta$, then
$$d_b(C)=w_b(\ccc)\geq w_{b-1}(\ccc)+1\geq d_{b-1}(C)+1.$$
Note that $d_{\theta}(C)$ can not equal $n$ since there exists at least one codeword in $C$ of $b$-symbol weight less than $n$.

Two linear codes $C_1$ and $C_2$ have the same $b$-symbol weight hierarchy if they have the same $\0$'s run distribution for any codeword. Therefore, $C$ and $CDP^{i}$ have the same $b$-symbol weight hierarchy.
\end{proof}

\subsubsection{Singleton Bound}

Wei \cite{Wei} established the generalized Singleton Bound in 1991.
\begin{lemma}\label{geneHamming}
{\rm\cite{Wei}}{\rm(Generalized Singleton Bound)} For an $[n,k,d]$ linear code over $\F_q$, $\dd_r(C)\leq n-k+r$ for $1\leq r\leq k.$
\end{lemma}
Wei called an $[n,k]$ code $C$ $r$-rank maximum distance separable ($r$-rank MDS for short) if $\dd_r(C)=n-k+r$.

Chee {\it et al}. \cite{CL,DTG} established the $b$-symbol Singleton Bound in 2013.
\begin{lemma}\label{bMDS}{\rm\cite{CL,DTG}}{\rm($b$-Symbol Singleton Bound)}
 If $C$ is an $(n,M,d_b(C))_q$ $b$-symbol code, then we have $M\leq q^{n-d_b(C)+b}$.
\end{lemma}
An $(n,M,d_b(C))_q$ $b$-symbol code $C$ with $M=q^{n-d_b(C)+b}$ is called a $b$-symbol maximum distance separable ($b$-symbol MDS for short) code.

Despite of the different metrics, their Singleton-type Bounds are the same. The only difference is that when the metric is the $r$-th generalized Hamming metric, we need to restrict $C$ to be linear.
\begin{lemma}\label{lem14}{\rm\cite{Wei,CL,DTG}}
Let $C$ be an MDS code with parameters $(n,q^k)$ over $\F_q.$ Then:
\begin{itemize}
  \item If $C$ is linear, then $C$ is an $r$-rank MDS code and $\dd_r(C)=d_H(C)+r-1$ for $1\leq r\leq k$ {\rm (see \cite{Wei})}.
  \item $C$ is a $b$-symbol MDS code and $d_b(C)=d_H(C)+b-1$ for $1\leq b\leq k$ {\rm(see\cite{CL,DTG})}.
\end{itemize}
\end{lemma}

The following two theorems give the characterization of trivial $r$-rank MDS codes or $b$-symbol MDS codes.

\begin{theorem}
A cyclic code $C$ with parameters $[n,k]_q$ is an $r$-rank MDS code if $r=k$.
\end{theorem}
\begin{proof}
By the definition of the $r$-th generalized Hamming metric, we have $\dd_r(C)=n$ since $C$ is cyclic. According to Lemma \ref{geneHamming}, $C$ is $r$-rank MDS.
\end{proof}
\begin{theorem}\label{Thm9}
A cyclic code $C$ with parameters $[n,k]_q$ is a $b$-symbol MDS  code if $b=k$. Moreover, $d_b(C)=n$ if $b\geq k$.
\end{theorem}
\begin{proof}
Every nonzero codeword in $C$ is generated by recursion of degree $k$ and thus has at most $k-1$ consecutive zeros. According to Lemma \ref{bMDS}, we get the desired result.
\end{proof}
\begin{remark}\label{re23}
(i) Theorem \ref{Thm9} proves the stability theorem in \cite{SOS} by using a more concise way and does not need the restriction $\gcd(n,q)=1$.

(ii) If $C$ is cyclic, we have $\theta= k-1$. In the sequel, we assume that $b$ is always less than $k$ if $C$ is cyclic. For the same cyclic code $C$, it has two weight hierarchies as follows:
\begin{itemize}
  \item generalized weight hierarchy: $\{\dd_1(C),\dd_2(C),\ldots,\dd_k(C)\}.$
  \item $b$-symbol weight hierarchy: $\{d_1(C),d_2(C),\ldots,d_k(C)\}.$
\end{itemize}
Further, we have $d_H(C)=d_1(C)=\dd_1(C)$ and $d_k(C)=\dd_k(C)=n.$
\end{remark}

\subsubsection{Griesmer Bound}
In 1992, Helleseth {\it et al.} \cite{HKY} established the generalized Griesmer Bound.
\begin{lemma}\label{ddr}{\rm \cite{HKY}}{\rm(Generalized Griesmer Bound)}
Let $C$ be an $[n,k]$ code over $\F_q$. Then

$$n\geq \dd_r(C)+\sum_{i=1}^{k-r}\left\lceil\frac{q-1}{q^i(q^r-1)}\dd_r(C)
\right\rceil,$$
and
$$(q^r-1)\dd_{r-1}(C)\leq (q^r-q)\dd_r(C).$$
\end{lemma}

We naturally expect to give the $b$-symbol Griesmer Bound for cyclic codes. There are two reasons why we restrict $C$ to being cyclic. If $C$ is cyclic, then
\begin{enumerate}
  \item the maximum $\mathbf{0}$'s run length of $C$ is less than the dimension of $C$;
  \item for any codeword $\ccc\in C$, the code generated by $G_b(\ccc)$ is a subcode of $C$.
\end{enumerate}
\begin{lemma}\label{shu}
Let $C$ be a cyclic code with parameters $[n,k]$ over $\F_q$.
Let $k=tb+s$, where $t, s$ are two nonnegative integers and $0\leq s<b$. Assume that any $b$-consecutive positions of $C$ are independent. Then
\begin{equation}\label{informationsym}
  \sum_{\ccc\in C}w_b(\ccc)=nq^{k-b}(q^b-1).
\end{equation}
Moreover, we have
\begin{eqnarray}\label{ferruh}
  d_b(C) &\leq& \left\lfloor\frac{q^k-q^{k-b}}{q^k-1}\cdot n\right\rfloor, \\
  n      &\geq&\left\lceil \sum_{i=0}^{t-1}\frac{d_b(C)}{q^{bi}}+
  \frac{d_b(C)}{q^{(t-1)b}}\cdot\frac{q^s-1}{q^s}\cdot\frac{1}
  {q^b-1}\right\rceil,\\
  n      &\geq&
  \left\lceil\sum_{i=0}^{t-1}\frac{d_b(C)}{(q^b)^i}\right\rceil
   \hbox {~~~if~~ $b|k$}.
\end{eqnarray}

\end{lemma}
\begin{proof}
The first claim follows by observing that any $b$-consecutive positions are information symbols implies that the cyclic code restricted to these positions contains any $b$-tuple equally often, i.e., $q^{k-b}$ times.

According Equality (\ref{informationsym}), we have
\begin{equation*}
  (q^k-1)d_b(C)\leq nq^{k-b}(q^b-1).
\end{equation*}
Then Inequality (10) holds, and
\begin{eqnarray*}
  n &\geq& \left\lceil d_b(C)\cdot\frac{(q^k-1)}
  {q^{k-b}(q^b-1)}\right\rceil \\
  ~ &=& \left\lceil d_b(C)\cdot\frac{(q^{tb+s}-1)}
  {q^{(t-1)b+s}(q^b-1)}\right\rceil \\
  ~ &=& \left\lceil d_b(C)\cdot\frac{(q^{tb}-1)}{q^{(t-1)b}(q^b-1)}
  +d_b(C)\cdot\frac{(q^s-1)}{q^{(t-1)b+s}(q^b-1)}\right\rceil \\
  ~ &=& \left\lceil\sum_{i=0}^{t-1}\frac{d_b(C)}{q^{bi}}
  +\frac{d_b(C)}{q^{(t-1)b}}\cdot\frac{q^s-1}{q^s}
  \cdot\frac{1}{q^b-1}\right\rceil
\end{eqnarray*}
Therefore, we obtain the desired result.
\end{proof}
The following lemma was given in \cite{SOS}. A nice relation between $d_b(C)$ and $d_{b-1}(C)$ was established, and it is the same as the relation between $\dd_r(C)$ and $\dd_{r-1}(C)$ which is given in Lemma \ref{ddr}.
\begin{lemma}{\rm\cite[Inequality (42)]{SOS}}\label{Ferruh13}
Let $C$ be a cyclic code with parameters $[n,k]$ over $\F_q$. Then
\begin{equation}\label{Ferruh12}
  (q^b-q)d_b(C)\geq (q^b-1)d_{b-1}(C).
\end{equation}
\end{lemma}
By Lemma \ref{Ferruh13}, we have
\begin{equation}\label{new4}
  d_b(C)\geq \left\lceil\frac{q^b-1}{q^b-q}d_{b-1}(C)\right\rceil.
\end{equation}
Yaakobi et al. proved
$d_H(C)\leq\left\lfloor\frac{2n}{3}\right\rfloor$ over $\F_2$ in \cite[Theorem 1]{Yaa1} to ensure that the lower bound on $d_2(C)$ does not exceed $n$.
 Notice that the lower bound on $d_b(C)$ given in (\ref{new4}) can not greater than $n$ since $d_{b-1}(C) \leq \left\lfloor\frac{q^k-q^{k-b+1}}{q^k-1}\cdot n\right\rfloor.$
Further, we have
\begin{equation}\label{new3}
  d_b(C)\geq \left\lceil\sum_{i=0}^{b-1}\frac{d_1(C)}{q^i}\right\rceil.
\end{equation}
In fact, the lower bound (\ref{new3}) can be further reinforced into the following result.
\begin{theorem}
Let $C$ be a cyclic code with parameters $[n,k]$ over $\F_q$. Then
\begin{equation}\label{new1}
  d_b(C)\geq\sum_{i=0}^{b-1}
  \left\lceil\frac{d_1(C)}{q^i}\right\rceil.
\end{equation}
\end{theorem}
\begin{proof}Let $\ccc$ be a codeword with $b$-symbol weight $w_b(\ccc)=d_b(C)$. By Lemma \ref{new7}, the rank of $G_b(\ccc)$ equals $b$.
Let $C_1$ be the linear code which generated by $G_b(\ccc)$. Since $Rank(G_b(\ccc))=b$, the parameters of $C_1$ are $[n,b,d_1(C_1)]$ with $d_1(C_1)\leq d_1(C)$. Assume that $$G_b(\ccc)=\left(
             \begin{array}{cccc}
               \mathbf{r}_1 & \mathbf{r}_2 & \cdots & \mathbf{r}_n \\
             \end{array}
           \right)_{b\times n},
$$
where $\mathbf{r}_i$ are the columns of $G_b(\ccc)$ with $1\leq i\leq n$. Let $$G_b^{\prime}(\ccc)=\left(
                                    \begin{array}{cccc}
                                      \mathbf{r}_{j_1} & \mathbf{r}_{j_2} & \ldots & \mathbf{r}_{j_m} \\
                                    \end{array}
                                  \right)_{b\times m},
$$ where $r_{j_1},\ldots,r_{j_m}$ are all the nonzero columns of $G_b(\ccc)$. Since the number of all the nonzero columns of $G_b(\ccc)$ equals $w_b(\ccc)$, $m=w_b(\ccc)$. Then the parameters of the linear code $C_2$ which generated by $G_b^{\prime}(\ccc)$ are $[w_b(\ccc),b,d_1(C_2)=d_1(C_1)\geq d_1(C)]$. By the Griesmer Bound, we have
\begin{equation*}
  d_b(C)=w_b(\ccc)\geq \sum_{i=0}^{b-1}\left\lceil\frac{d_1(C)}{q^i}\right\rceil.
\end{equation*}
Therefore, we obtain the desired result.
\end{proof}
Assume that $b|k$ and $k=tb$. By the preceding theorem, we have
\begin{equation}\label{new5}
  \sum_{i=0}^{t-1}
   \left\lceil\frac{d_b(C)}{q^{bi}}\right\rceil\geq
   \sum_{i=0}^{t-1}\left\lceil\frac{\sum_{j=0}^{b-1}
   \left\lceil\frac{d_1(C)}{q^j}\right\rceil}{q^{bi}}\right\rceil.
\end{equation}
From the Griesmer Bound, we have
\begin{equation}\label{new6}
  n\geq \sum_{i=0}^{tb-1}\left\lceil\frac{d_1(C)}{q^i}\right\rceil
  =\sum_{i=0}^{t-1}\sum_{j=0}^{b-1}\left\lceil\frac{d_1(C)}{q^{ib+j}}
  \right\rceil =\sum_{i=0}^{t-1}\sum_{j=0}^{b-1}\left\lceil\frac{\frac{d_1(C)}
  {q^j}}{q^{bi}}\right\rceil.
\end{equation}
Combining the Inequalities (\ref{new5}) and (\ref{new6}), a natural motivation for us is to explore the relation between $n$ and $\sum_{i=0}^{t-1}
   \left\lceil\frac{d_b(C)}{q^{bi}}\right\rceil.$ Then we have the following conjecture.

\begin{conjecture}\label{con1}{\rm($b$-Symbol Griesmer Bound for cyclic codes)}
Assume that $b|k$ and $k=tb$. If $C$ is a cyclic code with parameters $[n,k]$ over $\F_q$, then
\begin{equation}\label{conjecture1}
  n\geq\sum_{i=0}^{t-1}\left\lceil\frac{d_b(C)}{(q^b)^i}\right\rceil.
\end{equation}
\end{conjecture}
Inequality (\ref{conjecture1}) holds if $b=k$ or $b=1$ from Theorem \ref{Thm9} and the Griesmer Bound. The reader is warmly invited to attack Conjecture \ref{con1} when $b>1$ and $t>1$.


\section{Conclusions and open problems}
Many results in this paper generalize the previous work (e.g., Theorem \ref{relation}, Theorem \ref{thm5}).
In this paper, the $b$-symbol weight of a vector $\ccc$ is first calculated by calculating the $\mathbf{0}$'s run distribution of $\ccc$ (Theorem \ref{wbformula}). This provides a new way to calculate the $b$-symbol weight distribution of linear codes (e.g., in Remark \ref{r4} and Theorem \ref{single1} we calculated the $b$-symbol weight distribution of some special codes based on their $\mathbf{0}$'s run distributions).
 However, for general linear codes (or cyclic  codes), it is still a difficult task to determine their $\mathbf{0}$'s run distributions.
 A potential direction is to determine the $b$-symbol weight distributions of some special linear codes by exploring their $\mathbf{0}$'s run distributions.

Another highlight of this paper is that we first studied the connections and differences between $b$-symbol metric and $r$-th generalized Hamming metric. As two different generalizations of Hamming metric, they have a lot in common.
 When $C$ is cyclic, a very important relation between $d_b(C)$ and $\dd_b(C)$ is given (Theorem \ref{new8}).
 It is a pity that this paper fails to give the $b$-symbol Griesmer Bound for cyclic codes and only gives a conjecture (Conjecture \ref{con1}).
 It is important to note that $b$-symbol Griesmer bound does not apply to arbitrary linear codes (this is one of the differences from Griesmer Bound and generalized Griesmer Bound). However, the $b$-symbol Griesmer Bound is not restricted to cyclic codes; in fact, it is also applicable to constacyclic codes.
\section*{Acknowledgement}
This research is supported by Natural Science Foundation of China (12071001), Excellent Youth Foundation of Natural Science Foundation of Anhui Province (1808085J20), The Research Council of Norway under grant  (247742/O70).



\end{document}